\documentclass[a4paper, thm-restate, 11pt]{article}
\usepackage[bottom,para,symbol]{footmisc}
\usepackage{authblk}
\usepackage{amsmath}
\usepackage{amssymb}
\usepackage{amsthm}
\usepackage{graphicx}
\usepackage{complexity} 
\usepackage{enumitem}
\usepackage{algorithm}
\usepackage[noend]{algpseudocode}
\usepackage{comment}

\usepackage{footnote}
\algnewcommand{\LeftComment}[1]{\Statex \(\triangleright\) #1}

\usepackage[svgnames]{xcolor}
\usepackage{color}

\newtheorem{mytheorem}{Theorem}
\newtheorem{mydefinition}[mytheorem]{Definition}
\newtheorem{mylemma}[mytheorem]{Lemma}

\algrenewcommand\algorithmicrequire{\textbf{Input:}}
\algrenewcommand\algorithmicensure{\textbf{Output:}}

\newcommand{\decprob}[4]{
\vspace{2mm}
\noindent\fbox{
\begin{minipage}{0.94\textwidth}
\begin{tabular*}{\textwidth}{@{\extracolsep{\fill}}lr} \textsc{#1}  & {\bf{ }} #2
\\ \end{tabular*}
{\bf{Input:}} #3  \\
{\bf{Question:}} #4
\end{minipage}
}
\vspace{2mm}
}

\newcommand{\optprob}[4]{
\vspace{2mm}
\noindent\fbox{
\begin{minipage}{1\textwidth}
\begin{tabular*}{\textwidth}{@{\extracolsep{\fill}}lr} \textsc{#1}  & {\bf{ }} #2
\\ \end{tabular*}
{\bf{Input:}} #3  \\
{\bf{Output:}} #4
\end{minipage}
}
\vspace{2mm}
}

\usepackage{graphicx}
\usepackage{listings}
\usepackage{hyperref}
\usepackage{thmtools}
\usepackage{thm-restate}

\usepackage{pdfpages}
\usepackage{epigraph}

\newtheorem{reduction rule}{Reduction Rule}[section]

\date{}

\title{Color-Constrained Arborescences in Edge-Colored Digraphs\thanks{This work is supported by SERB CRG grant number CRG/2022/007751.}}
\author{P.~S.~Ardra\footnote{111914001@smail.iitpkd.ac.in}~}
\author{Jasine~Babu\footnote{jasine@iitpkd.ac.in}~}
\author{R.~Krithika\footnote{krithika@iitpkd.ac.in}~}
\author{Deepak~Rajendraprasad\footnote{deepak@iitpkd.ac.in}~}
\affil{Indian Institute of Technology Palakkad, Palakkad, India.}

\begin{document}
\maketitle
\begin{abstract}
Given a multigraph $G$ whose edges are colored from the set $[q]:=\{1,2,\ldots,q\}$ (\emph{$q$-colored graph}), and a vector $\alpha=(\alpha_1,\ldots,\alpha_{q}) \in \mathbb{N}^{q}$ (\emph{color-constraint}), a subgraph $H$ of $G$ is called \emph{$\alpha$-colored}, if $H$ has exactly $\alpha_i$ edges of color $i$ for each $i \in[q]$. In this paper, we focus on $\alpha$-colored arborescences (spanning out-trees) in $q$-colored multidigraphs. We study the decision, counting and search versions of this problem. It is known that the decision and search problems are polynomial-time solvable when $q=2$ and that the decision problem is \NP-complete when $q$ is arbitrary. However the complexity status of the problem for fixed $q$ was open for $q > 2$. We show that, for a $q$-colored digraph $G$ and a vertex $s$ in $G$, the number of $\alpha$-colored arborescences in $G$ rooted at $s$ for all color-constraints $\alpha \in \mathbb{N}^q$ can be read from the determinant of a symbolic matrix in $q-1$ indeterminates. This result extends Tutte's matrix-tree theorem for directed graphs and gives a polynomial-time algorithm for the counting and decision problems for fixed $q$.  We also use it to design an algorithm that finds an $\alpha$-colored arborescence when one exists. Finally, we  study the weighted variant of the problem and give a polynomial-time algorithm (when $q$ is fixed) which finds a minimum weight solution. 
\end{abstract}

\section{Introduction}
\epigraph{
\emph{Colors must fit together as pieces in a puzzle or cogs in a wheel.}}
{Hans Hofmann, Painter}

Many textbook problems in graph theory are on the existence of subgraphs of a specific kind -- triangles, cycles, spanning trees, perfect matchings, Hamiltonian cycles, and so on. A similar, albeit shorter, list of problems entice us on directed graphs too -- directed cycles, arborescences, directed Hamiltonian cycles. The corresponding algorithmic questions -- existence, counting, finding, enumerating -- serve as common invitations to the study of graph algorithms. One way to spice up these problems is to color the edges of the input graph and look for solutions in which the colors ``fit together''. More precisely, we assume that the edges of the input graph (undirected or directed, simple or otherwise) are colored with a palette $[q]$ of colors and that a color-constraint $\alpha \in \mathbb{N}^q$ is specified. We then restrict the solution-space to those subgraphs which are $\alpha$-colored, that is, those which have exactly $\alpha_i$ edges of color $i$ for each $i \in [q]$. We call these color-constrained subgraph problems.

These problems take a special flavor when the subgraphs we are looking for turn out to be bases of a matroid (spanning tree) or common bases of two or more matroids (perfect matching in bipartite graphs, arborescences in digraphs). The color-constraint we formalized above is also matroidal -- the edge-set of the input graph forms the ground set, the color-classes partition the ground set, and the color-constraint gives the capacity constraint for a partition matroid. Hence adding the color-constraint adds one more matroidal dimension to the problem. In particular, a color-constrained spanning tree is a common basis of two matroids while a color-constrained bipartite perfect  matching and a color-constrained arborescence are common bases of three matroids. Finding a largest common independent set of $k$ matroids on the same ground set is the \textsc{$k$-Matroid Intersection} problem.  A greedy algorithm solves it efficiently when $k=1$. There are many general polynomial-time algorithms to solve the case when $k=2$ and more efficient specialized algorithms for common special cases like bipartite matching and arborescence~\cite{Schrijver03}. The problem is \NP-hard for $k \geq 3$ in general. However, this does not rule out the existence of polynomial-time algorithms for special cases.

Color-constrained subgraph problems of this matroidal nature have received considerable attention since 1980s~\cite{BarahonaP87,GabowT84,Gusfield84,PapadimitriouY82,RendlL89}. 
Work on finding an $\alpha$-colored spanning tree in a $q$-colored undirected graph was started by Papadimitriou and Yannakakis in 1982 with a solution for the case of $q=2$~\cite{PapadimitriouY82}. They employed an elegant exchange procedure that starts with an efficiently-computable spanning tree, performs a sequence of ``edge swaps'' and ends with the required solution. Gabow and Tarjan~\cite{GabowT84} observed that a similar ``exchange sequence'' exists in the more general setting of matroids and used it to find a min-weight $\alpha$-colored spanning tree of a weighted $2$-colored graph in 1984. An alternative algorithm for the same problem was given by Gusfield~\cite{Gusfield84} at around the same time. Later, Rendl and Leclerc gave a polynomial-time algorithm for arbitrary number of colors~\cite{RendlL89} in 1989. The problem of finding a $(k, n/2 - k)$-colored perfect matching in an $n$-vertex graph, often called \textsc{Exact Perfect Matching}, was also born out of the same 1982 paper by Papadimitriou and Yannakakis \cite{PapadimitriouY82}. This one turned out to be much tougher to tame than its sibling. Papadimitriou and Yannakakis~\cite{PapadimitriouY82} conjectures that the problem is \NP-hard. Mulmuley, Vazirani and Vazirani \cite{mulmuley1987matching} gave a randomized polynomial-time algorithm for it in 1987. Despite several subsequent works on the problem, it still remains as one of those rare examples of a problem in \RP\ not yet known to be in \P. However, finding a color-constrained perfect matching in a $q$-colored bipartite graph is \NP-hard if $q$ is unbounded~\cite{itai1977some,tanimoto1978some}.

 In this work, we are interested in finding and counting $\alpha$-colored arborescences rooted at a given vertex $s$ in $q$-colored digraphs. Before describing our work further, we make the decision variant of the problem precise. For reasons that will become clear soon, we only specify the constraints on the first $q-1$ colors. Given a color-constraint $\alpha=(\alpha_1,\ldots,\alpha_{q-1}) \in \mathbb{N}^{q-1}$, an arborescence of a $q$-colored $n$-vertex digraph is said to be an {\em $\alpha$-colored $s$-arborescence} if it is rooted at $s$ and has exactly $\alpha_i$ edges of color $i$ for each $i \in[q-1]$. The number of edges of color $q$ is automatically constrained.

\decprob{\textsc{Color-Constrained Arborescence (CC-Arb)}}{}{A $q$-colored $n$-vertex multidigraph $G$, a vertex $s\in V(G)$ and a color-constraint $\alpha \in \mathbb{N}^{q-1}$.}{Does $G$ have an $\alpha$-colored $s$-arborescence?}

Since \textsc{CC-Arb} is a special case of $3$-matroid intersection, we do not have any off-the-matroid-shelf algorithm for the problem. In 1987, Barahona and Pulleyblank~\cite{BarahonaP87} gave a polynomial-time algorithm for the case of $q=2$. Recently,  Ardra et al. \cite{ArdraNKKPR24} showed that the problem is \NP-complete for arbitrary $q$. The complexity status of the problem for fixed $q$ was open for $q > 2$. 

\medskip\noindent 
\textbf{Results.} We show that, for every fixed $q$, \textsc{CC-Arb} is polynomial-time solvable by extending Tutte's matrix-tree theorem for directed multigraphs to $q$-colored directed multigraphs (Theorem~\ref{thm:ext-Tutte}). This theorem also allows us to efficiently count the number $\alpha$-colored $s$-arborescences. With some additional work, we can also use it to  find in polynomial time an $\alpha$-colored $s$-arborescence when one exists. Finally, we show that our methods can be extended to the case of finding a min-weight $\alpha$-colored $s$-arborescence in a weighted $q$-colored multidigraph.




\medskip
\noindent \textbf{Terminology.} We refer to the book by Diestel~\cite{Diestel17} for standard graph-theoretic definitions and terminology not defined here. For a digraph $G$, let $V(G)$ and $E(G)$ denote the sets of its vertices and edges/arcs, respectively.  The {\em out-degree} (resp., {\em in-degree}) of a vertex $v$ is the number of edges leaving (resp. entering) $v$. For a digraph $G$, let $\overleftarrow{G}$ denote the graph obtained from $G$ by reversing its arcs. 

We assume that the input $q$-colored digraph is in the adjacency matrix representation whose entries are elements of $\mathbb{N}^q$ which encode the number of edges of each color. For a $q$-colored digraph $G$ and a vertex $s \in V(G)$, the $q$-colored digraph obtained from $G$ by adding a self-loop of color $q$ at $s$ is denoted as $G^s$. For vertices $i,j$ and color $c \in [q]$, let $d_{ijc}$ denote the number of $c$-colored edges  from $i$ to $j$. 

A {\em functional digraph} is a digraph in which every vertex has out-degree exactly one. 
A spanning functional subgraph of an $n$-vertex digraph has exactly $n$ edges. Given a color-constraint $\alpha=(\alpha_1,\ldots,\alpha_{q-1}) \in \mathbb{N}^{q-1}$, a spanning functional subgraph of a $q$-colored digraph is said to be $\alpha$-colored if it has exactly $\alpha_i$ edges from $E_i$ for each $i \in[q-1]$.

\section{Color-Sensitive Matrix-Tree Theorems}

Two classical Laplacian matrices associated with a directed multigraph $G$ on the vertex-set $[n]$ are the out-degree Laplacian $L_G$ and the in-degree Laplacian $L'_G$ defined as follows.

\begin{equation}
\label{eqn:out-Laplacian}
    L_G[i,j] = 
    \begin{cases}
        -d_{ij},                & \text{ if } i\neq j,  1\leq i,j\leq n\\ 
        \sum_{k=1}^n d_{ik},    & \text{ if } i=j,  1\leq i\leq n
    \end{cases},   
\end{equation}
and
\begin{equation}
\label{eqn:in-Laplacian}
    L'_G[i,j] = 
    \begin{cases}
        -d_{ji},                & \text{ if } i\neq j,  1\leq i,j\leq n\\ 
        \sum_{k=1}^n d_{ki},    & \text{ if } i=j,  1\leq i\leq n
    \end{cases},   
\end{equation}
where $d_{ij}$ is the number of arcs from vertex $i$ to vertex $j$. It is easy to verify that the in-degree Laplacian of $G$ is same as the out-degree Laplacian of $\overleftarrow{G}$ and vice-versa.

Tutte's directed matrix-tree theorem~\cite{Tutte48} states that, if $G$ is a digraph without self-loops, then the determinant of the submatrix obtained by deleting the $i$-th row and column of the out-degree Laplacian $L_G$ (resp. the in-degree Laplacian $L'_G$) is the number of spanning in-trees (resp. out-trees) of $G$ rooted at $i$. Even though our results are stated on spanning out-trees, we will work with the out-degree Laplacian and in-trees in the proof of Theorem~\ref{thm:ext-Tutte} so that we can use the language of functional digraphs. Instead of doing the row and column deletions on $L_G$, one can also add a self-loop at $i$ to (the otherwise loopless) $G$ and then compute the determinant of the out-degree Laplacian of the resulting digraph $G^i$. One can verify this equivalence easily by seeing the new Laplacian $L_{G^i}$ as the result of deleting the $(n+1)$-th row and column of the Laplacian of a supergraph of $G$ obtained by adding a new vertex $n+1$ and a new arc $(i, n+1)$ to $G$. A similar observation holds good for the in-degree Laplacian too. 

In order to count color-constrained arborescences in $q$-colored digraphs, we define two symbolic matrices in $q-1$ indeterminates to take the role of the two Laplacians.

\begin{mydefinition}
    For a $q$-colored multidigraph $G$ on the vertex-set $[n]$, the symbolic out-degree Laplacian matrix $L_G$ and the symbolic in-degree Laplacian matrix $L'_G$ are
    \begin{equation}
    \label{eqn:symOutLaplacian}
    L_G[i,j]= 
    \begin{cases}
        -\sum_{c=1}^q d_{ijc}~x_{c}, & 
            \text{ if } i\neq j,  1\leq i,j\leq n\\ 
        \sum_{k=1}^n \sum_{c=1}^q d_{ikc}~x_{c}, 
            & \text{ if } i=j,  1\leq i\leq n
    \end{cases},    
    \end{equation}
    and
    \begin{equation}
    \label{eqn:symInLaplacian}
    L'_G[i,j]= 
    \begin{cases}
        -\sum_{c=1}^q d_{jic}~x_{c}, & 
            \text{ if } i\neq j,  1\leq i,j\leq n\\ 
        \sum_{k=1}^n \sum_{c=1}^q d_{kic}~x_{c}, 
            & \text{ if } i=j,  1\leq i\leq n
    \end{cases},    
    \end{equation}
where $x_1,\ldots,x_{q-1}$ are indeterminates, $x_q=1$ and $d_{ijc}$ is the number of $c$-colored arcs from vertex $i$ to vertex $j$.
\end{mydefinition}

When each $x_c$ is set to $1$, $L_{G}$ collapses to the out-degree Laplacian of the underlying uncolored digraph. It is also interesting that we can decompose $L_G$ as
\begin{equation}
    L_G = \sum_{c = 1}^q L_{G_c} x_c,
\end{equation}
where $L_{G_c}$ is the classical out-degree Laplacian of the spanning subgraph $G_c$ of $G$ consisting of the $c$-colored edges of $G$. A similar collapse and decomposition applies to the symbolic in-degree Laplacian too. 

The course is now set to introduce the work-horse of this paper.

\begin{mytheorem}[Extended Tutte's Theorem]
\label{thm:ext-Tutte}
    Given a $q$-colored loopless multidigraph $G$, a vertex $s \in V(G)$ and a color-constraint $\alpha \in \mathbb{N}^{q-1}$, the number of $\alpha$-colored $s$-arborescences of $G$ is the coefficient of the monomial $\prod_{c=1}^{q-1} x_c^{\alpha_c}$ in the determinant polynomial of the submatrix obtained by deleting the row and column corresponding to $s$ from the symbolic in-degree Laplacian $L'_G$ of $G$.
\end{mytheorem}

The original Tutte's theorem has many proofs~\cite{borchardt1860ueber,chaiken1978matrix,de2020elementary,Tutte48}, the first of which predates that of Tutte by nearly a century~\cite{borchardt1860ueber}. A long but instructive way to prove Theorem~\ref{thm:ext-Tutte} would be to adapt one of the elementary proofs of Tutte's theorems which expands the determinant of $L_{\overleftarrow{G^s}}$ (where $G^s$ is $G$ with a self-loop added to $s$) using the  Leibniz formula where each non-zero term corresponds to a product over the arcs of a functional subgraph of $\overleftarrow{G^s}$ and then showing that all terms corresponding to functional subgraphs with non-singleton cycles cancel. However, in the interest of space, we derive our result from the following generalization of Tutte's theorem by Moon~\cite{Moon94}.

\begin{mytheorem}[Corollary~4.1 in \cite{Moon94}]
\label{thm:Moon}
    Let $M$ be an $n\times n$ matrix defined as \\
    \begin{equation}
    \label{eqn:MoonM}
      M[i,j]=\begin{cases}
            -z_{ij} & \text{ if } i\neq j, 1\leq i,j\leq n\\ 
            \overset{n}{\underset{k=1}{\sum}} z_{ik} & \text{ if } i=j, 1\leq i\leq n\\
        \end{cases},
    \end{equation}
    where $z_{ij}$ are indeterminates for each $i,j \in [n]$. 
    Then 
    \begin{equation}
    \label{eqn:MoonDet}
        \det(M) = \sum_{D\in \mathcal{D}} \prod_{(i,j) \in E(D)} z_{ij},      
    \end{equation}
    where $\mathcal{D}$ is the set of all functional digraphs $D$ with $V(D)=[n]$ such that each cycle of $D$ has length $1$. 
\end{mytheorem}

The role of our symbolic Laplacians in counting is first illustrated for functional subgraphs rather than for arborescences so that the key idea is not lost in self-loops and direction reversals. The shorthand $x^{\alpha}$ will be used to denote the monomial $\prod_{c=1}^{q-1} x_c^{\alpha_c}$ henceforth.

\begin{mylemma}
\label{lem:nFuncDigraphs}
    Given a $q$-colored multidigraph $G$ on the vertex set $[n]$ and a color-constraint $\alpha \in \mathbb{N}^{q-1}$, the coefficient of the monomial $x^{\alpha}$  in the determinant polynomial of the symbolic out-degree Laplacian $L_G$ is the number of $\alpha$-colored spanning functional subgraphs $F$ of $G$ such that each cycle of $F$ has length $1$.
\end{mylemma}
\begin{proof}
    If we set $z_{ij} = \sum_{c=1}^q d_{ijc}x_c$  with $x_q = 1$ for all $i,j \in [n]$ (including $i=j$), then the matrix $M$ in (\ref{eqn:MoonM}) is equal to the symbolic out-degree Laplacian $L_G$ of $G$ (Eqn.~(\ref{eqn:symOutLaplacian})). Hence Theorem~\ref{thm:Moon} gives
    \begin{equation}
        \det(L_G) = 
            \sum_{D \in \mathcal{D}} P_D,    
    \end{equation}   
    where $\mathcal{D}$ is the set of all functional digraphs $D$ with $V(D)=[n]$ such that each cycle of $D$ has length $1$ and
    \begin{equation}
        P_D=
            \prod_{(i,j) \in E(D)} 
                \sum_{c=1}^q d_{ijc}x_c.
    \end{equation}   
    
    The polynomial $P_D$ only contains terms coming from the set $\mathcal F$ of spanning functional subgraphs of $G$ whose underlying uncolored digraph is $D$. Each $\alpha$-colored $F \in \mathcal F$ adds a one to the coefficient of $x^{\alpha}$ in the polynomial $P_D$ and this is the only contribution of $F$ to $P_D$. Hence the the coefficient of the monomial $x^{\alpha}$ in $P_D$ is the number of $\alpha$-colored spanning functional subgraphs of $G$ whose underlying uncolored digraph is $D$. Now it is easy to see that the coefficient of the monomial $x^{\alpha}$ in $\det(L_G)$ is the total number of $\alpha$-colored spanning functional subgraphs $F$ of $G$ such that each cycle of $F$ has length 1. 
\end{proof}

We now show how to use Lemma~\ref{lem:nFuncDigraphs} to count arborescences instead of functional subgraphs. Given a $q$-colored loopless digraph $G$, a vertex $s \in V(G)$ and a color-constraint $\alpha \in \mathbb{N}^{q-1}$, let $\mathcal{T}$ denote the set of all $\alpha$-colored $s$-arborescences in $G$. Let $\overleftarrow{G^s}$ be the digraph obtained by adding a self-loop of color $q$ at the vertex $s$ to $G$ and reversing the direction of all the arcs. Let $\mathcal{F}$ denote the set of all $\alpha$-colored spanning functional subgraphs $F$ of $\overleftarrow{G^s}$ such that all cycles in $F$ have length one. It is easy to see that for any $T \in \mathcal{T}$, the graph $\psi(T)$ obtained by reversing all the arcs of $T$ and adding a self-loop of color $q$ at $s$ is in $\mathcal{F}$. Moreover, $\psi$ is a bijection since all the functional digraphs $F \in \mathcal{F}$ contain only one self-loop -- the $q$-colored self-loop at $s$ and hence there is a unique $T \in \mathcal{T}$ obtained by reversing the arcs of $F$ and deleting the self-loop at $s$. Hence $|\mathcal{T}| = |\mathcal{F}|$. This allows us to count $|\mathcal{T}|$ using Lemma~\ref{lem:nFuncDigraphs} and thus we have the next lemma. Recall that $L_{\overleftarrow{G}} = L'_G$ for a digraph $G$.

\begin{mylemma}
\label{lem:nArb}
    Let $G$ be a $q$-colored loopless multidigraph on the vertex set $[n]$, $s \in V(G)$ and $\alpha \in \mathbb{N}^{q-1}$. Let $G^s$ denote the multigraph obtained by adding a self-loop of color $q$ at the vertex $s$. Then the coefficient of the monomial $x^{\alpha}$ in the determinant polynomial of the symbolic in-degree Laplacian $L'_{G^s}$ is the number of $\alpha$-colored $s$-arborescences in $G$.
\end{mylemma}

Since the incoming arcs at $s$ do not appear in any $s$-arborescence, we can as well assume that there are no incoming arcs at $s$ in $G$. Hence the $s$-th row of $L'_{G^s}$ has $x_q$ (which is $1$) at the $s$-th position (due to the self-loop added) and $0$ everywhere else. Hence the determinant of $L'_{G^s}$ is equal to the determinant of the submatrix of $L'_G$ obtained by deleting the $s$-th row and column. This completes the proof of Theorem~\ref{thm:ext-Tutte}.

We end this section with the undirected analogue of Theorem~\ref{thm:ext-Tutte}. Towards this, we define the symbolic Laplacian of a $q$-colored undirected multigraph as follows.

\begin{mydefinition}
    For a a $q$-colored multigraph $G$ on the vertex-set $[n]$, the symbolic Laplacian matrix $L_G$ is
    \begin{equation}
    \label{eqn:symLaplacian}
    L_G[i,j]= 
    \begin{cases}
        -\sum_{c=1}^q d_{ijc}~x_{c}, & 
            \text{ if } i\neq j,  1\leq i,j\leq n\\ 
        \sum_{k=1}^n \sum_{c=1}^q d_{ikc}~x_{c}, 
            & \text{ if } i=j,  1\leq i\leq n
    \end{cases},    
    \end{equation}
    where $x_1,\ldots,x_{q-1}$ are indeterminates, $x_q=1$ and $d_{ijc}$ is the number of $c$-colored arcs between vertices $i$ and $j$.
\end{mydefinition}

Given a $q$-colored undirected multigraph $G$, let $D$ denote the digraph obtained from $G$ by replacing every edge $e = \{u,v\}$ of $G$ by the pair of edges $(u,v)$, $(v,u)$ of the same color as that of $e$. It is easy to verify that, once we fix a vertex $r$ in $G$, there is a bijection between the set of all spanning trees of $G$ and the set of all $r$-arborescences of $D$. Now applying Theorem~\ref{thm:ext-Tutte} on $D$ with $r$ as the source vertex, we obtain the following result. 

\begin{mytheorem}[Extended Kirchhoff's Theorem]
\label{thm:ext-Kirchoff}
    Given a $q$-colored loopless multigraph $G$ and a color-constraint $\alpha \in \mathbb{N}^{q-1}$, the number of $\alpha$-colored spanning trees of $G$ is the coefficient of the monomial $\prod_{c=1}^{q-1} x_c^{\alpha_c}$ in the determinant polynomial of the submatrix obtained by deleting the first row and column of the symbolic Laplacian $L_G$ of $G$.
\end{mytheorem}

\section{Algorithms for Unweighted Digraphs}

Theorem~\ref{thm:ext-Tutte}, our workhorse, immediately gives us a way to solve \textsc{CC-Arb} and its counting variant. The time-complexity for both is that of computing the determinant of an $(n-1) \times (n-1)$ matrix where each entry is a linear polynomial in $q-1$ indeterminates. Once we do that, we can count the number of $\alpha$-colored $s$-arborescences in the given $q$-colored graph for all legal $\alpha$ in one shot.

There are multiple algorithms that compute the determinant of a matrix of multivariate polynomials - like Gaussian Elimination, Expansion by Minors, Characteristic Polynomial Method and Evaluate-Interpolate Method -- whose running time depends differently on the size of the matrix, the total number of indeterminates, the degree of each polynomial and also on the sparsity of the matrix. Horowitz and Sahni~\cite{HorowitzS75} has a done a theoretical and empirical comparison of these methods. Some of the methods perform much better in practice than others and some of them can also be nicely parallelized \cite{marco2004parallel}. Since we are in the regime where the number of indeterminates $q-1$ is a fixed quantity and each indeterminate appears with a degree one in each entry of the matrix, the Evaluate-Interpolate method with the following time-complexity is a good choice.

\begin{mytheorem}{\em \cite{HorowitzS75}}
\label{thm:horowitz-sahni}
    Let $M$ be an $n\times n$ matrix whose entries are polynomials in $r$ variables $x_1,\ldots,x_r$ and each entry of $M$ has degree at most $d$ in each of the variables. Then the Evaluate-Interpolate Algorithm computes $\det(M)$
    using $\mathcal{\widetilde{O}}(((d+1)n)^{r}(n^{3}+(d+1)n^{2}))$ arithmetic operations.
\end{mytheorem}

Horowitz and Sahni~\cite{HorowitzS75} analyze the number of coefficient multiplications needed to compute $\det(M)$ and explicitly ignore the time taken for merging the partial sums. We use $\mathcal{\widetilde{O}}$ notation to account for this cost since it at most adds a poly-logarithmic factor. 
As our algorithms are algebraic, we first analyze the running time in terms of the number of arithmetic operations on the coefficients of the polynomials involved. Though each polynomial in the input matrix has coefficients upper bounded by $m$, the subsequent computations may result in coefficients that are as large as $(2mn)^n$ which makes single-precision arithmetic infeasible even for moderate input size and we will consider this issue separately.


\begin{mylemma}
    \label{lemma:deterministic-find-arb}
    Given a $q$-colored multidigraph $G$ on $n$ vertices and $m$ edges, a vertex $s \in V(G)$ and a color-constraint $\alpha \in \mathbb{N}^{q-1}$, we can count the number of $\alpha$-colored $s$-arborescences of $G$ using $\mathcal{\widetilde{O}}(2^q \cdot n^{q+2})$ multi-precision operations. 
    Moreover, if there is at least one $\alpha$-colored $s$-arborescence in $G$, then we can find one using $\mathcal{\widetilde{O}}(m \cdot 2^q \cdot n^{q+2})$ multi-precision operations.
\end{mylemma}

\begin{proof}
    Given $G$, relabeling its vertices with labels from $[n]$ and constructing $L'_G$ takes $\mathcal{O}(n^2q)$ time. Let $M$ be the sub-matrix of $L'_G$ obtained by deleting the $s$-th row and column. Observe that each entry of $M$ has degree at most $1$ in each of the $q-1$ variables. Theorem~\ref{thm:horowitz-sahni} guarantees that the Evaluate-Interpolate method takes $\mathcal{\widetilde{O}}(2^q \cdot n^{q+2})$ operations. Once the determinant is obtained, a simple linear scan over all monomials, which are $\mathcal{O}(n^q)$ in number, gives us the coefficient of $x^{\alpha}$. Therefore, the number of $\alpha$-colored $s$-arborescences in $G$ can be determined using $\mathcal{\widetilde{O}}(2^q \cdot n^{q+2})$ operations. 

    Next we describe how to find a solution when one exists. We iterate through each edge $e$ of $G$, removing it from $G$ if $G \setminus e$ has an $\alpha$-colored $s$-arborescence. Since the remaining graph still contains an $s$-arborescence after each removal, the procedure ends with a subgraph $G'$ of $G$ which has an $s$-arborescence. If $G'$ has even one more edge than an $s$-arborescence, then the earliest such edge to be processed in the iteration would have been deleted. Hence $G'$ is itself an $s$-arborescence. This procedure takes $\mathcal{\widetilde{O}}(m \cdot 2^q \cdot n^{q+2})$ operations. Note that we can remove any duplicate edges (parallel edges of the same color) before the search so that $m$ is at most $n^2q$.
\end{proof}

When single-precision arithmetic is infeasible due to blow-up in the size of coefficients, Horowitz and Sahni~\cite{HorowitzS75} recommends solving the problem using modular arithmetic. We choose a sequence of single-precision odd-primes $p_1 < \cdots< p_k$, such that $p_1$ is larger than the maximum number of distinct points needed for any intermediate polynomial evaluation, which is $2n$ in our case, and $\prod_{i=1}^k p_i$ to be larger than the largest coefficient in the final answer which is at most $m^n$ in our case. Hence it suffices to pick $n$ single-precision primes larger than $\max\{m,2n\}$ and repeat the determinant computations over each $\mathbb{Z}_{p_i}$. The time to convert the coefficients in the input matrix to $\mathbb{Z}_{p_i}$ can be ignored since every coefficient is at most $m$ which is less than $p_i$. The time for final reconstruction is $\mathcal{O}(n^2)$ using the Chinese Remainder algorithm~\cite{Gerhard13}. Thus if the number of multi-precision arithmetic operations performed by our algorithm is $t(m,n,q)$ then the overall running time is  $\mathcal{O}(n \cdot t(m,n,q)+ n^2)$ single-precision arithmetic operations. This gives the main result of this section. 


%


\begin{mytheorem}
    \label{thm:deterministic-find-arb}
    Given a $q$-colored multidigraph $G$ on $n$ vertices and $m$ edges, a vertex $s \in V(G)$ and a color-constraint $\alpha \in \mathbb{N}^{q-1}$, we can count the number of $\alpha$-colored $s$-arborescences of $G$ in $\mathcal{\widetilde{O}}(2^q \cdot n^{q+3})$ time. 
    Moreover, if there is at least one $\alpha$-colored $s$-arborescence in $G$, then we can find one in $\mathcal{\widetilde{O}}(m \cdot 2^q \cdot n^{q+3})$ time.
\end{mytheorem}

\section{Algorithms for Weighted Digraphs} 

In this section, we bring weighted $q$-colored digraphs to limelight and search for a min-weight arborescence. In particular, we study the following problem.


\optprob{\textsc{Minimum Weight Color-Constrained Arborescence(Min CC-Arb)}}{}{A $q$-colored multidigraph $G$ with a weight function $w:E(G) \rightarrow \mathbb{Z}^+$, a vertex $s\in V(G)$ and a color-constraint $\alpha \in \mathbb{N}^{q-1}$.}{A min-weight $\alpha$-colored $s$-arborescence of $G$.}

Note that we can handle negative integral weights by adding a constant bias to all the edges and we restrict the range of weights to $\mathbb Z^+$ for easier reading. Since our interest is limited to finding a min-weight solution, we can get rid of duplicate edges in $G$. Indeed, if we have parallel edges $(i,j)$ of the same color, we need to retain only a min-weight edge among them in $G$. We will henceforth assume that our input graph does not have duplicate edges. In the terminology of the previous section, this would amount to saying $d_{ijc} \in \{0,1\}$. We define the weighted symbolic version of the Laplacian as follows.

\begin{mydefinition}
    For a $q$-colored multidigraph $G$ on the vertex-set $[n]$ without duplicate edges and a weight function $w: E(G) \to \mathbb{Z}^+$, the weighted symbolic out-degree Laplacian matrix $L_{G,w}$ and the weighted symbolic in-degree Laplacian matrix $L'_{G,w}$ are
    \begin{equation}
    \label{eqn:wtSymOutLaplacian}
    L_{G,w}[i,j]= 
    \begin{cases}
        -\sum_{c=1}^q w_{ijc}~x_{c}, & 
            \text{ if } i\neq j,  1\leq i,j\leq n\\ 
        \sum_{k=1}^n \sum_{c=1}^q w_{ikc}~x_{c}, 
            & \text{ if } i=j,  1\leq i\leq n
    \end{cases},    
    \end{equation}
    and
    \begin{equation}
    \label{eqn:wtSymInLaplacian}
    L'_{G,w}[i,j]= 
    \begin{cases}
        -\sum_{c=1}^q w_{jic}~x_{c}, & 
            \text{ if } i\neq j,  1\leq i,j\leq n\\ 
        \sum_{k=1}^n \sum_{c=1}^q w_{kic}~x_{c}, 
            & \text{ if } i=j,  1\leq i\leq n
    \end{cases},    
    \end{equation}
where $x_1,\ldots,x_{q-1}$ are indeterminates, $x_q=1$ and $w_{ijc}$ is the weight of the $c$-colored arc from vertex $i$ to vertex $j$ if one exists and $0$ otherwise.
\end{mydefinition}

For any subgraph $H$ of $G$, $w^{\times}(H)$ will denote $\prod_{e \in E(H)} w(e)$.
One can figure out every coefficient of $\det(L_{G,w})$ by following the proof of Lemma~\ref{lem:nFuncDigraphs}. In the earlier case, each $\alpha$-colored spanning functional subgraph $F$ of $G$ with only singleton cycles contributed a one to the coefficient of $x^{\alpha}$. In the present case, $F$ will contribute $w^{\times}(F)$ to the coefficient of $x^{\alpha}$.  Hence the weighted version of Lemma~\ref{lem:nFuncDigraphs} is the following.


\begin{mylemma}
\label{lem:nWtFuncDigraphs}
    Given a $q$-colored multidigraph $G$ on the vertex set $[n]$ without duplicate edges, a weight function $w: E(G) \to \mathbb{Z}^+$ and a color-constraint $\alpha \in \mathbb{N}^{q-1}$, the coefficient of the monomial $x^{\alpha}$  in the determinant polynomial of the weighted symbolic out-degree Laplacian $L_{G,w}$ is $\sum_{F \in \mathcal{F}_{\alpha}} w^{\times}(F)$, where $\mathcal{F}_{\alpha}$ is the set of $\alpha$-colored spanning functional subgraphs $F$ of $G$ such that each cycle of $F$ has length one.
\end{mylemma}

From Lemma~\ref{lem:nWtFuncDigraphs} one can prove this weighted variant of Theorem~\ref{thm:ext-Tutte} by following the same steps as in the previous section.

\begin{mytheorem}
\label{thm:ext-WtTutte}
    Given a $q$-colored multidigraph $G$ on the vertex-set $[n]$ without self-loops and duplicate edges, a weight function $w: E(G) \to \mathbb{Z}^+$, a vertex $s \in V(G)$ and a color-constraint $\alpha \in \mathbb{N}^{q-1}$, the coefficient of the monomial $x^{\alpha}$ in the determinant polynomial of the submatrix obtained by deleting the $s$-th row and column of the weighted symbolic in-degree Laplacian $L'_{G,w}$ is $\sum_{T \in \mathcal{T}_{\alpha}} w^{\times}(T)$, where $\mathcal{T}_{\alpha}$ is the set of $\alpha$-colored $s$-arborescences $G$.
\end{mytheorem}

Now we describe an attempt to find the minimum weight of an $\alpha$-colored $s$-arborescence in the given digraph. For any prime $r$, let $w_{r}$ denote a new weight function defined by $w_{r}(e) = {r}^{w(e)}$ for all $e \in E(G)$. Let $c_{\alpha, r}$ denote the coefficient of $x^{\alpha}$ in the determinant of the submatrix of $L'_{G,w_r}$ obtained by deleting the $s$-th row and column. Theorem~\ref{thm:ext-WtTutte} shows that $c_{\alpha, r} = \sum_{T \in \mathcal{T}_{\alpha}} r^{w(T)}$, where $\mathcal{T}_{\alpha}$ is the set of $\alpha$-colored $s$-arborescences in $G$ and $w(T) = \sum_{e \in E(T)} w(e)$ is the total weight of $T$ under the original weight function $w$. 
This idea is already used in \cite{mulmuley1987matching} to find a matching of a given weight and in \cite{BarahonaP87} to find an arborescence of a given weight. 

If there is a unique min-weight arborescence $T^*$ in $\mathcal{T}_{\alpha}$, then we can extract $w(T^*)$ from $c_{\alpha, r}$. In this case, $c_{\alpha, r}$ will be the sum of $r^{w(T^*)}$ and strictly larger powers of $r$. Hence $w(T^*)$ is the largest $k$ for which $r^k$ is a factor of $c_{\alpha, r}$. But this method will fail if the number of min-weight arborescences in $\mathcal{T}_{\alpha}$ is a multiple of $r$. One way to solve this is by scaling up the original weights by a factor of $2mn$ and then perturbing the resulting weights by adding a number chosen uniformly at random from $[2m]$. The Isolation Lemma \cite{mulmuley1987matching} guarantees that with probability at least $1/2$, the min-weight solution is unique. And hence we can extract $w(T^*)$ from $c_{\alpha,2}$. But this results in a randomized algorithm. Hence we propose an alternate method to tackle the non-uniqueness.

Our method relies on the observation that we do not need uniqueness of the min-weight solution for the extraction of the minimum weight from $c_{\alpha, r}$ to work. It is sufficient that the number of min-weight solutions is not a multiple of $r$. Let $w_{min}$ be the minimum weight of a solution and $t$ be the number of min-weight solutions in $\mathcal{T}_{\alpha}$. One can verify that $\max\{k : r^k \mid c_{\alpha, r}\} \geq w_{min}$ and the equality holds if and only if $t$ is not a multiple of $r$. Let $r_1, \ldots, r_n$ be $n$ distinct single-precision primes larger than $m$. Since $\prod_{i=1}^n r_i > m^n$, no integer $t \in [m^n]$ is a multiple of all the primes in this list. We compute $c_{\alpha, r_i}$ and then $f(r_i) = \max\{k : r_i^k \mid c_{\alpha, r_i}\}$ for each $i \in [n]$.  Let $r_j$ be a prime in the list $r_1, \ldots, r_n$ that does not divide $t$. Hence $f(r_j) = w_{min}$. For every other prime $r_i$ in the list, $f(r_i) \geq w_{min}$. Hence, even though we do not know $r_j$ in advance, we can extract $w_{min}$ as $\min\{f(r_i) : i \in [n]\}$. Once we know $w_{min}$, we can find a min-weight solution by iterating through all the edges $e$ in $E(G)$ and deleting $e$ from $G$ if the minimum weight of a solution in $G \setminus e$ is also $w_{min}$. This procedure will terminate with a min-weight solution. The next lemma is the analogue of Lemma~\ref{lemma:deterministic-find-arb} for the weighted case.


\begin{mylemma}
    \label{lem:minCCarb}
    Given a $q$-colored multidigraph $G$ on $n$ vertices and $m$ edges, a weight function $w: E(G) \to [W]$ with $W \in n^{\mathcal{O}(1)}$, a vertex $s \in V(G)$ and a color-constraint $\alpha \in \mathbb{N}^{q-1}$, we can find the minimum weight of an $\alpha$-colored $s$-arborescence in $G$ using $\mathcal{\widetilde{O}}(2^q \cdot n^{q+3})$ multi-precision operations. 
    Moreover,  we can find a min-weight solution using $\mathcal{\widetilde{O}}(m \cdot 2^q \cdot n^{q+3})$ multi-precision operations.
\end{mylemma}

Since $r^k$ can be computed in $\mathcal{O}(\log k)$ multi-precision operations, the time to construct $L'_{G, w_{r}}$ is $\mathcal{O}(n^2q \log W)$. Let $M$ be the sub-matrix of $L'_{G,w_r}$ obtained by deleting the $s$-th row and column. Each entry of $M$ has degree at most $1$ in each of the $q-1$ variables and Theorem~\ref{thm:horowitz-sahni} guarantees that the Evaluate-Interpolate method takes $\mathcal{\widetilde{O}}(2^q \cdot n^{q+2})$ multi-precision operations. Once the determinant is obtained, a linear scan over all monomials gives us the coefficient $c_{\alpha,r}$ of $x^{\alpha}$. Note that $c_{\alpha, r} \leq m^n r^{nW}$. Therefore, computing $f(r)$ takes $\mathcal{O}(\log (n W))$ divisions. Since we compute $L'_{G,w_r}$  and $f(r)$ for each of the $n$ primes, the total number of multi-precision operations to find $w_{min}$ is $\mathcal{\widetilde{O}}(n(2^q \cdot n^{q+2} + \log (n W)))$ which is claimed in Lemma~\ref{lem:minCCarb} since $W \in n^{\mathcal{O}(1)}$. The time complexity of finding a  min-weight solution is $m$ times the above. This completes the time complexity analysis of Lemma~\ref{lem:minCCarb}.

Now, we will analyze the time complexity in single-precision operations using modular arithmetic and Chinese Remaindering. Notice that now we have to work with much larger coefficients than in the unweighted case. Each entry of $L'_{G,w_r}$ has coefficients upper bounded by $r^W$ and the coefficients in the resulting determinant polynomial are at most $m^n r^{nW}$. This time, we pick $nW$ primes $p_1,  p_2, \ldots, p_{nW},$ each larger than $\max\{2n,m r\}$, repeat the determinant computations over each $p_i$. This blows up the running time by a factor of $nW$. The final reconstruction of the coefficient of $x^\alpha$ takes $\mathcal{O}(n^2 \cdot W^2)$ time. Hence each determinant computation for each $r$ takes $\mathcal{\widetilde{O}}(n \cdot W \cdot 2^q \cdot n^{q+2}+n^2 \cdot W^2)$ single-precision operations. The final running time is hence $n$ times this. 
 

\begin{mytheorem}
    \label{thm:minCCarb}
    Given a $q$-colored multidigraph $G$ on $n$ vertices and $m$ edges, a weight function $w: E(G) \to [W]$ with $W \in n^{\mathcal{O}(1)}$, a vertex $s \in V(G)$ and a color-constraint $\alpha \in \mathbb{N}^{q-1}$, we can find the minimum weight of an $\alpha$-colored $s$-arborescence in $G$ in $\mathcal{\widetilde{O}}(2^q \cdot n^{q+4} \cdot W +n^3 \cdot W^2)$ time. 
    Moreover,  we can find a minimum weight solution in $\mathcal{\widetilde{O}}(m(2^q \cdot n^{q+4} \cdot W +n^3 \cdot W^2))$ time. 
\end{mytheorem}

\section{Concluding Remarks}
The results of this paper straightaway pose two interesting open problems. The first is fixed-parameter tractability of \textsc{CC-Arb} with $q$ as the parameter. Our work shows that \textsc{CC-Arb} is in \XP\ when parameterized by $q$. The second one is the enumeration version of \textsc{CC-Arb}. Enumerating all $\alpha$-colored $s$-arborescences with polynomial delay (for fixed $q$) is an interesting problem. The algorithm that we have given for the search version of \textsc{CC-Arb} does not straightaway extend to an enumeration algorithm.

Though \textsc{CC-Arb} and \textsc{Exact Perfect Matching} in bipartite graphs reduce to \textsc{3-Matroid Intersection}, there seems to be a sharp contrast in their complexities. The former is in \P\ for every fixed $q$ while the latter is in \RP\ but not known to be in \P\ even for $q=2$. \textsc{CC-Arb} is \textsc{3-Matroid Intersection} where two of the three matroids are partition matroids and the third one is a graphic matroid. \textsc{Exact Perfect Matching} restricted to bipartite graphs is \textsc{3-Matroid Intersection} where all the three matroids are partition matroids. It is intriguing whether the hardness of the problems is dependent on this difference. A deep dive into this is tempting.  

We also hope that our colored extensions of the classic matrix-tree theorems of Tutte and Kirchhoff will open another bridge between algebra and parameterized algorithms. In fact, a fixed-parameter tractable algorithm for \textsc{CC-Arb} with $q$ as the parameter can give practical algorithms to compute the symbolic determinant of a subclass of large matrices of linear polynomials when the number of indeterminates is small.


\bigskip \noindent \textbf{Acknowledgement. }We are grateful to K. Murali Krishnan and Piyush P. Kurur for their comments and suggestions.

\bibliography{ref}

\end{document}